\documentclass[graybox,footinfo]{svmult}

\smartqed
\usepackage{mathptmx}       
\usepackage{helvet}         
\usepackage{courier}        
\usepackage{type1cm}        
\usepackage{graphicx}       

\usepackage{array,colortbl}
\usepackage{amsmath,amsfonts,amssymb,bm} 

\newcommand{\rd}{\,\mathrm{d}} 
\newcommand{\bst}{\boldsymbol{t}}    

\usepackage{microtype} 

\usepackage[colorlinks=true,linkcolor=black,citecolor=black,urlcolor=black]{hyperref}
\usepackage{bookmark}
\newtheorem{algorithm}{Algorithm}

\makeatletter 
  \providecommand*{\toclevel@author}{999}
  \providecommand*{\toclevel@title}{0}
\makeatother

\begin{document}

\title*{Discrepancy Estimates for  Acceptance-Rejection Samplers Using Stratified Inputs }
\author{Houying Zhu \and Josef Dick}
\institute{
Houying Zhu \and Josef Dick
\at School of Mathematics and Statistics,  University of New South Wales, Sydney, NSW, Australia
\email{houying.zhu@student.unsw.edu.au}, \email{josef.dick@unsw.edu.au}
%
}
\maketitle

\abstract{In this paper we propose an acceptance-rejection sampler using stratified inputs as  diver sequence. We estimate the discrepancy of the points generated by this algorithm.
First we show an upper bound on the star discrepancy of order  $N^{-1/2-1/(2s)}$. Further we prove  an upper bound on the $q$-th moment of the $L_q$-discrepancy $(\mathbb{E}[N^{q}L^{q}_{q,N}])^{1/q}$  for $2\le q\le \infty$, which is of order $N^{(1-1/s)(1-1/q)}$. We also  present an improved convergence rate for a deterministic acceptance-rejection algorithm using $(t,m,s)-$nets as driver sequence.}

\section{Introduction}
\label{sec:1}
Markov chain Monte Carlo (MCMC) sampling is a classical method widely used in  simulation.
Using a deterministic sequence as driver sequence in the MCMC procedure, known as Markov chain quasi-Monte Carlo (MCQMC) algorithm,  shows potential to improve the convergence rate.
Tribble and Owen~\cite{TribbleOwen2005} proved a consistency result for MCMC estimation for finite state spaces.
 A~construction of weakly completely uniformly distributed
sequences is also proposed. As a sequel to
the work of Tribble, Chen  \cite{ChenThesis2011} and Chen,
Dick and Owen~\cite{CDO2011} demonstrated that MCQMC algorithms using a completely uniformly distributed
sequence as driver sequence give a consistent result under certain assumptions on
the update function and Markov chain. Further, Chen~\cite
{ChenThesis2011} also showed that \mbox{MCQMC} can achieve
a convergence rate of $O(N^{-1+\delta})$ for any $\delta>0$ under
certain stronger assumptions, but he only showed the existence of a driver sequence.

In a different direction, L'Ecuyer, Lecot and Tuffin \cite{Pierre2008}  proposed a randomized quasi-Monte Carlo method which simulates multiple Markov chains  in parallel and
randomly permutes  the driver sequence in order to reduce variance.   Garber and Choppin in \cite{GC2014} adapted low
discrepancy point sets instead of random numbers in sequential Monte
Carlo (SMC). They proposed a new algorithm, named sequential quasi-Monte
Carlo (SQMC), through the use of a Hilbert  space-filling curve.
They constructed consistency and stochastic bounds based on randomized
QMC point sets for this algorithm.
More literature review about applying QMC to MCMC problems can be found
in \cite[Section 1]{CDO2011} and the references therein.

In  \cite{DRZ2013}, jointly done with Rudolf, we prove upper
bounds on the discrepancy for uniformly ergodic Markov chains driven by
a deterministic sequence rather than independent random
variables. We
show that there exists a deterministic driver sequence such that the discrepancy of the Markov
chain from the target distribution with respect to certain test sets converges with almost the usual Monte Carlo rate
of $N^{-1/2}$. In the  sequential work \cite{DR2013} done by Dick and Rudolf, they  consider  upper
bounds on the discrepancy under the assumption that the Markov chain
is  variance bounding and the driver sequence is deterministic. In particular, they proved a better existence result, showing a   discrepancy bound having a rate
of convergence of almost $N^{-1}$ under a  stronger assumption on the update function, the so called anywhere-to-anywhere condition.

The acceptance-rejection algorithm is one of the  widely used techniques for sampling from a distribution when direct simulation is not possible or expensive. The idea of this method is to determine a good choice of proposal density (also known as hat function), then sample from the proposal density with low cost. In particular, Devroye \cite{Devroye1984} gave a construction method of a proposal density for  log-concave densities and H\"ormann \cite{Hormann1995} proposed a rejection procedure, called transformed density rejection, to  construct a proposal density.  Detailed summaries of this technique and some extensions can be found in the monographs  \cite{BHL2013} and \cite{HLD04}. For many target densities, finding a good proposal density is difficult.  An alternative approach to improve efficiency is to determine  a better choice of sequences having the  designated proposal density.

 The deterministic acceptance-rejection algorithm has  been
discussed by Moskowitz and Caflisch~\cite{MC96}, Wang~\cite{Wang1999, Wang00} and Nguyen and \"{O}kten \cite{NO2013}, where  empirical evidence or a consistency result were given. Two measurements included therein are the empirical root mean square error (RMSE) and the empirical standard deviation. However,  the
discrepancy of samples is not directly investigated. Motivated by those papers, in \cite{ZD2014} we investigated  the discrepancy properties of points produced by
a totally deterministic acceptance-rejection method. We proved that the discrepancy of
samples generated by a QMC acceptance-rejection sampler is bounded
from above by $N^{-1/s}$. A lower bound shows that for any given driver
sequence, there always exists a target density such that the star
discrepancy is at most $N^{-2/(s+1)}$.

In this work we first present an acceptance-rejection algorithm using stratified inputs as  driver sequence. Stratified sampling is one of the variance reduction methods  used in Monte Carlo sampling. More precisely,  grid-based stratified  sampling improves  the RMSE to
$N^{-1/2-1/s}$ for Monte Carlo, see for instance \cite[Chapter 10]{OwenbookDraft}. In this paper, we are  interested in the discrepancy properties of points produced by the acceptance-rejection method  with stratified inputs as driver sequence. We obtain a convergence rate of the star-discrepancy of order $N^{-1/2-1/(2s)}$. Also an estimation of the $L_q$-discrepancy is considered for this setting. One would expect that the convergence rate which can be achieved using deterministic sampling methods also depends on properties of the target density function. One such  property is the number of elementary intervals (for a precise definition see Definition~\ref{Einterval} below) of a certain size needed to cover the graph of the density. We show that if the graph can be covered by a small number of elementary intervals, then an improved rate of convergence can be achieved using $(t,m,s)$-nets as driver sequence. In general, this strategy does not work with stratified sampling, unless one knows the elementary 
intervals explicitly.

The paper is organized  as follows. In Section~\ref{Sec:2} we provide the needed notation and background.  Section~\ref{Sec:3} introduces the proposed acceptance-rejection sampler using stratified inputs, where an existence upper bound on the star-discrepancy and an estimation  of the $L_q$-discrepancy  are given.  Section~\ref{Sec:4} illustrates  an improved rate of convergence when using $(t,m,s)$-nets as driver sequences.
\section{Preliminaries}\label{Sec:2}
We are interested in the discrepancy properties of samples generated by the acceptance-rejection sampler. We consider the  $L_q$-discrepancy and the star-discrepancy.
\begin{definition}[{$L_q$-discrepancy}]\label{lqdiscrepancy}  Let $1\le q\le \infty$ be a real number. For a point set $P_N$ in $[0,1]^{s}$, the $L_q$-discrepancy is defined by
\begin{equation*}
  L_{q,N}=\Big(\int_{[0,1]^s}\big|\frac{1}{N}\sum_{n=0}^{N-1}1_{[\boldsymbol{0},\boldsymbol{t})}(\boldsymbol{x}_n)-\lambda([\boldsymbol{0},\boldsymbol{t}))\big|^q \D \boldsymbol{t}\Big)^{1/q},
  \end{equation*} where $1_{[\boldsymbol{0},\boldsymbol{t})}(\boldsymbol{x}_n)=\left\{\begin{array}{ll}
1,&  \mbox{if }\boldsymbol{x}_n\in [\boldsymbol{0},\boldsymbol{t}) ,\\
0,& \mbox{otherwise}.\\
\end{array}
\right.$, $[\boldsymbol{0},\boldsymbol{t})=\prod_{j=1}^{s}[0, t_j)$ and $\lambda$ is the Lebesgue measure, with the obvious modification for $q=\infty$. The $L_{\infty,N}$-discrepancy is called the star-discrepancy which is also denoted by $D_N^*(P_N)$.
\end{definition}

The acceptance-rejection algorithm accepts all points below the graph of the density function. In order to prove bounds on the discrepancy, we assume that the set below the graph of the density function admits a so-called Minkowski content. We introduce the Minkowski content in the following. For a set $A$ we denote the boundary of $A$ by $\partial A$.
\begin{definition}[Minkowski content] For a set $A\subseteq [0,1]^s$, let
\begin{equation*}
  \mathcal{M}(\partial A)=\lim_{\varepsilon\to 0}\frac{\lambda((\partial A)_{\varepsilon})}{2\varepsilon},
  \end{equation*} where $(\partial A)_{\varepsilon}=\{\boldsymbol{x}\in\mathbb{R}^s| \| \boldsymbol{x}-\boldsymbol{y} \|\le \varepsilon \mbox{ for } \boldsymbol{y}\in \partial A \}$ and $\|\cdot\|$ denotes the Euclidean norm. If $\mathcal{M}(\partial A)$ (abbreviated as $\mathcal{M}_A$ without causing confusion) exists and is finite, then $\partial A $ is said to admit an $(s-1)-$dimensional Minkowski content.
\end{definition}

\section{Acceptance-Rejection Sampler Using Stratified Inputs }
\label{Sec:3} We now present the acceptance-rejection algorithm using stratified inputs.
\begin{algorithm}\label{algorithmSAR} Let the target density $\psi:[0,1]^{s-1}\to \mathbb{R}_+$ where $s\geq 2$, be given. Assume that there exists a constant $L<\infty$ such that $\psi(\boldsymbol{z}) \leq L$ for  all $\boldsymbol{z}\in[0,1]^{s-1}$. Let $A=\{\boldsymbol{z}\in[0,1]^s:\psi(z_1,\ldots,z_{s-1})\geq L z_s\}$  and assume that $\partial A$ admits an $(s-1)-$dimensional Minkowski content.
\begin{itemize}
  \item [i)~] Let $M\in\mathbb{N}$ and let $\{Q_0,\ldots,Q_{M-1}\}$ be a disjoint covering with of $[0,1]^s$ with $Q_i$ of the form $\prod_{j=1}^{s}\left[\frac{c_j}{M^{1/s}},\frac{c_j+1}{M^{1/s}}\right)$ with $0\le c_j\le \lceil M^{1/s}\rceil-1$. Then $\lambda(Q_i)=1/M$ for all $0\le i\le M-1$. Generate a point set $P_M=\{\boldsymbol{x}_0,\ldots, \boldsymbol{x}_{M-1}\}$  such that there is exactly one point of $P_M$ uniformly distributed in each sub-cube $Q_i$.
  \item [ii)~]Use the acceptance-rejection method for the points in $P_M$ with respect to the density $\psi$, i.e. we accept the point $\boldsymbol{x}_n$ if $\boldsymbol{x}_n \in A$, otherwise reject. Let $P_N^{(s)}=A\cap P_M=\{\boldsymbol{z}_0,\ldots, \boldsymbol{z}_{N-1}\}$ be the sample set we accept.
  \item [iii)~]  Project the points we accepted $P_N^{(s)}$ onto  the first $(s-1)$ coordinates. Let
    $Y_N^{(s-1)}=\{\boldsymbol{y}_0,\ldots, \boldsymbol{y}_{N-1}\}$ be the projections of the points  $P_N^{(s)}=\{\boldsymbol{z}_0,\ldots, \boldsymbol{z}_{N-1}\}$.
 \item [iv)~] Return the point set $Y_N^{(s-1)}$.
\end{itemize}
\end{algorithm}

\subsection{Existence Result of Samples with Small Star Discrepancy}
Here we present some results that we will use to prove an upper bound for the star discrepancy with respect to points generated by the acceptance-rejection sampler using stratified inputs.  For any $0\le\delta\le 1$, a set $\Gamma$ of anchored boxes $[\boldsymbol{0},\boldsymbol{x}]\subseteq [0,1]^{s}$ is called a $\delta$-cover of the set of anchored boxes $[\boldsymbol{0},\boldsymbol{t}]\subseteq [0,1]^{s}$ if for every point $\boldsymbol{t}\in [0,1]^s$, there exist $[\boldsymbol{0},\boldsymbol{x}],[\boldsymbol{0},\boldsymbol{y}]\in \Gamma$ such that $[\boldsymbol{0},\boldsymbol{x}]\subseteq [\boldsymbol{0},\boldsymbol{t}]\subseteq [\boldsymbol{0},\boldsymbol{y}]$ and $\lambda( [\boldsymbol{0},\boldsymbol{y}]\setminus [\boldsymbol{0},\boldsymbol{x}])\le \delta$. The following result on the size of the $\delta$-cover is obtained from \cite[Theorem 1.15]{G2008}.

\begin{lemma}\label{delta_coverOpt}For any $s$ and $\delta$ there exists a $\delta$-cover of the set of anchored boxes $[\boldsymbol{0},\boldsymbol{t}]\subseteq [0,1]^{s}$  which has cardinality at most $(2e)^s(\delta^{-1}+1)^s$.\end{lemma}

By a simple generalization, the following result holds for our setting.

\begin{lemma}\label{delta_coverOptA} Let $\psi:[0,1]^{s-1}\to \mathbb{R}_+$, where $s\geq 2$, be a function. Assume that there exists a constant $L<\infty$ such that $\psi(\boldsymbol{z}) \leq L$ for  all $\boldsymbol{z}\in[0,1]^{s-1}$. Let
 $A=\{\boldsymbol{z}\in[0,1]^s:\psi(z_1,\ldots,z_{s-1})\geq L z_s\}$ and $J_{\boldsymbol{t}}^*=([\boldsymbol{0},\boldsymbol{t})\times [0,1])\cap A $.
Let $(A, \mathcal{B}(A), \lambda)$ be a probability space where $\mathcal{B}(A)$
is the Borel $\sigma$-algebra of $A$. Define the set $\mathcal{A} \subset \mathcal{B}(A)$
of test sets by
$$\mathcal{A} = \{J_{\boldsymbol{t}}^*: \boldsymbol{t} \in [0,1]^{s-1} \}.$$
Then for any $\delta > 0$ there exists a $\delta$-cover $\Gamma_\delta$ of $\mathcal{A}$
with $$|\Gamma_\delta| \le (2e)^{s-1}(\delta^{-1}+1)^{s-1}.$$
\end{lemma}

\begin{proof} Let
\begin{equation*}
  \Gamma_\delta:=\{ ([\boldsymbol{0},\boldsymbol{x}]\times [0,1])\cap A, [\boldsymbol{0},\boldsymbol{x}] \in \Gamma \},
\end{equation*} where $\Gamma$ is a $\delta$-cover of the set of anchored boxes $[\boldsymbol{0},\boldsymbol{t}]\subseteq [0,1]^{s-1}$ with $|\Gamma|\le (2e)^{s-1}(\delta^{-1}+1)^{s-1}$. By Lemma~\ref{delta_coverOpt} such a $\delta$-cover $\Gamma$ exists.
For any set $J_t^*\in \mathcal{A}$, there exist $([\boldsymbol{0},\boldsymbol{x}]\times [0,1])\cap A, ([\boldsymbol{0},\boldsymbol{y}]\times [0,1])\cap A \in \Gamma_\delta$ such that
$$ ([\boldsymbol{0},\boldsymbol{x}]\times [0,1])\cap A \subseteq J_t^*\subseteq ([\boldsymbol{0},\boldsymbol{y}]\times [0,1])\cap A,$$
and $$\lambda\Big(\big(([\boldsymbol{0},\boldsymbol{y}]\times [0,1])\cap A \big )\setminus \big(([\boldsymbol{0},\boldsymbol{x}]\times [0,1])\cap A \big)\Big)\le\lambda( [\boldsymbol{0},\boldsymbol{y}]\setminus [\boldsymbol{0},\boldsymbol{x}])\le \delta. $$

Hence $ \Gamma_\delta$ forms a $\delta$-cover of $\mathcal{A}$ and $|\Gamma_\delta|=|\Gamma|$.
\end{proof}

\begin{lemma}\label{J_Intersect} Let the unnormalized  density function $\psi:[0,1]^{s-1}\to \mathbb{R}_+$, with $s\geq 2$, be given. Assume that  there exists a  constant $L< \infty$ such that $\psi(\boldsymbol{z})\le L$ for all $\boldsymbol{z}\in [0,1]^{s-1}$.
 \begin{itemize}
   \item  Let $M\in \mathbb{N}$ and let the disjoint subsets $Q_0,\ldots,Q_{M-1}$ be of the form  $\prod_{i=1}^{s}\left[\frac{c_j}{M^{1/s}},\frac{c_j+1}{M^{1/s}}\right)$ where $0\le c_j\le \lceil M^{1/s}\rceil-1 $. These sets form a disjoint covering of  $[0,1]^s$ and each set $Q_i$ satisfies $\lambda(Q_i)=1/M$.
   \item  Let
$$A=\{\boldsymbol{z}\in[0,1]^s:\psi(z_1,\ldots,z_{s-1})\geq L z_s\}.$$
Assume that $\partial A $  admits an $(s-1)-$dimensional Minkowski content $\mathcal{M}_A$.
   \item  Let $J_{\boldsymbol{t}}^*=([\boldsymbol{0},\boldsymbol{t})\times [0,1])\bigcap A $, where $\boldsymbol{t}=(t_1,\ldots,t_{s-1})\in[0,1]^{s-1}$.
 \end{itemize}
Then there exists an $M_0\in\mathbb{N}$ such that $J_{\boldsymbol{t}}^*$ at most intersects with $3s^{1/2}\mathcal{M}_AM^{1-1/s}$ subcubes $Q_i$ for all $M\ge M_0$.
\end{lemma}
The result can be obtained utilizing  a similar proof as in \cite[Theorem~4.3]{{HeOwen2014}}. For the sake of completeness, we repeat the proof here.

\begin{proof} Since $\partial A $  admits an $(s-1)-$dimensional Minkowski content, it follows that
\begin{equation*}
  \mathcal{M}_A=\lim_{\varepsilon\to 0}\frac{\lambda((\partial A)_{\varepsilon})}{2\varepsilon} <\infty.
\end{equation*} Thus by the definition of the limit, for any fixed $\vartheta>2$, there exists $\varepsilon_0$ such that
$ \lambda((\partial A)_{\varepsilon})\le \vartheta \varepsilon \mathcal{M}_{A}$ whenever $\varepsilon\le \varepsilon_0$.

Based on the form of the subcube given by $\prod_{i=1}^{s}\left[\frac{c_j}{M^{1/s}},\frac{c_j+1}{M^{1/s}}\right)$, the largest diagonal length is $\sqrt{s}M^{-1/s}$. We can assume that
$M> (\sqrt{s}/\varepsilon_0)^s$, then $\sqrt{s}M^{-1/s}=:\varepsilon<\varepsilon_0$ and $\bigcup_{i\in J}Q_i\subseteq (\partial A)_{\varepsilon}$, where
$J$ is the index set for the sets $Q_i$ which satisfy $Q_i\cap A\neq \emptyset$. Therefore
\begin{equation*}
  |J|\le \frac{\lambda((\partial A)_{\varepsilon})}{\lambda(Q_i)}\le \frac{\vartheta \varepsilon  \mathcal{M}_A }{M^{-1}}=\sqrt{s}\vartheta\mathcal{M}_AM^{1-1/s}.
\end{equation*} Without loss of generality, we can set $\vartheta=3$, which completes the proof.
\end{proof}

\begin{remark}  Ambrosio et al \cite{ACV2008} found that for a closed set $A\subset \mathbb{R}^s$, if $A$ has a Lipschitz boundary, then $\partial A$ admits an $(s-1)$-dimensional Minkowski content. In particular,  a convex set $A\subset [0,1]^s$ has an $(s-1)$-dimensional Minkowski content. Note that the  surface area of a convex set in $[0,1]^s$ is bounded by the surface area of the unit cube $[0,1]^s$, which is $2s$ and it was also shown by Niederreiter and Wills \cite{NW75} that $2s$  is best possible. It follows that the Minkowski content $\mathcal{M}_A\le 2s$ when $A$ is a convex set in $[0,1]^s$.
\end{remark}

\begin{lemma}\label{BoundN} Suppose that  all the assumptions of Lemma~\ref{J_Intersect} are satisfied. Let $N$ be the number of points accepted by Algorithm~\ref{algorithmSAR}. Then we have
\begin{equation*}
  M(\lambda(A)-{3s^{1/2}\mathcal{M}_A}{M^{-1/s}})\le N\le M(\lambda(A)+{3s^{1/2}\mathcal{M}_A}{M^{-1/s}}).
\end{equation*} \end{lemma}
\begin{proof}The number of points we accept in Algorithm~\ref{algorithmSAR} is a random number since the driver sequence given by stratified inputs  is random. Let $\mathbb{E}(N)$ be the expectation of $N$.  The number of $Q_i$ which have non-empty intersection with $A$ is bounded by $l=3s^{1/2}\mathcal{M}_AM^{1-1/s}$ from Lemma~\ref{J_Intersect}.
Thus
\begin{equation}\label{NandM}
  \mathbb{E}[N]-l\le N\le \mathbb{E}[N]+l.
\end{equation} Further we have
\begin{equation}\label{ExpectN}
  \mathbb{E}[N]=\sum_{i=0}^{M-1}\frac{\lambda(Q_i\cap A)}{\lambda(Q_i)}=M\lambda(A).
\end{equation} Combining \eqref{NandM} and \eqref{ExpectN} and substituting $l=3s^{1/2}\mathcal{M}_AM^{1-1/s}$, one obtains the desired result.
\end{proof}

Before we start to prove the upper bound on the star-discrepancy, our method requires the well-known Bernstein-Chernoff inequality.
\begin{lemma}\label{Bernstein-Chernoff}\cite[Lemma~2]{Beck1984} Let $\eta_0,\ldots,\eta_{l-1}$ be independent random variables with $\mathbb{E}(\eta_i)=0$ and $|\eta_i|\le 1$ for all $0\le i\le l-1$. Denote by $\sigma_i^2$ the variance of $\eta_i$, i.e. $\sigma_i^2=\mathbb{E}(\eta_i^2)$. Set $\beta=(\sum_{i=0}^{l-1}\sigma_i^2)^{1/2}$. Then for any $\gamma>0$ we have
\begin{equation*}\mathbb{P}\Big(\big|\sum_{i=0}^{l-1}\eta_i\big|\ge \gamma\Big)\le\left\{\begin{array}{ll}
2e^{-\gamma/4},&  \mbox{if }\gamma\ge\beta^2,\\
2e^{-\gamma^2/4\beta^2},& \mbox{if }\gamma\le\beta^2.\\
\end{array}
\right.\end{equation*}
\end{lemma}

\begin{theorem} \label{Bound_Stratified}
Let an unnormalized  density function $\psi:[0,1]^{s-1}\to \mathbb{R}_+$, with $s\geq 2$, be given.  Assume that  there exists a  constant $L< \infty$ such that $\psi(\boldsymbol{z})\le L$ for all $\boldsymbol{z}\in [0,1]^{s-1}$. Let $C=\int_{[0,1]^{s-1}}\psi(\boldsymbol{z}) \rd\boldsymbol{z}$ and let
 the graph under $\psi$  be defined as
$$A=\{\boldsymbol{z}\in[0,1]^s:\psi(z_1,\ldots,z_{s-1})\geq L z_s\}.$$
Assume that $\partial A $  admits an $(s-1)-$dimensional Minkowski content $\mathcal{M}_A$.
 Then for  all large enough $N$, with positive probability, Algorithm~\ref{algorithmSAR} yields a  point set $Y_N^{(s-1)}\subseteq [0,1]^{s-1}$
such that
\begin{align*}
D_{N,\psi}^*(Y_N^{(s-1)})\le \frac{s^{\frac{3}{4}}\sqrt{6\mathcal{M}_A}}{(2L)^{\frac{1}{2s}-\frac{1}{2}}C^{\frac{1}{2}-\frac{1}{2s}}}\frac{\sqrt{\log N}}{N^{\frac{1}{2}+\frac{1}{2s}}}+\frac{2C}{LN}.
  \end{align*}
\end{theorem}

\begin{proof}
 Let $J_{\boldsymbol{t}}^*=([\boldsymbol{0},\boldsymbol{t})\times [0,1])\bigcap A $, where $\boldsymbol{t}=(t_1,\ldots,t_{s-1})$. Using the notation from Algorithm~\ref{algorithmSAR},
let $\boldsymbol{y}_n $ be the first $s-1$ coordinates  of $\boldsymbol{z}_n \in A$. For $ n=0,\ldots,N-1$, we have $$\sum_{n=0}^{M-1}1_{J_{\boldsymbol{t}}^*}(\boldsymbol{x}_n)=\sum_{n=0}^{N-1}1_{[\boldsymbol{0},\boldsymbol{t})}(\boldsymbol{y}_n).$$
 Therefore
 \begin{equation}\label{Discre}
   \Big|\displaystyle\frac{1}{N}\sum_{n=0}^{N-1}1_{[\boldsymbol{0},\boldsymbol{t})}({\boldsymbol{y}_n})
    -\displaystyle\frac{1}{C}\int_{[\boldsymbol{0},\boldsymbol{t})}\psi({\boldsymbol{z}})d\boldsymbol{z}\Big|=\Big|\displaystyle\frac{1}{N}\sum_{n=0}^{M-1}1_{J_{\boldsymbol{t}}^*}({\boldsymbol{x}_n})
    - \frac{1}{\lambda(A)} \lambda(J_{\boldsymbol{t}}^*) \Big|.\end{equation}
It is noted that
\begin{eqnarray}\label{DiscreTrans}
\Big|\sum_{n=0}^{M-1}1_{J_{\boldsymbol{t}}^*}({\boldsymbol{x}_n})
    - \frac{N}{\lambda(A)} \lambda(J_{\boldsymbol{t}}^*) \Big| &
\le& \Big| \sum_{n=0}^{M-1} 1_{J_{\boldsymbol{t}}^*}(\boldsymbol{x}_n) - M\lambda(J_{\boldsymbol{t}}^*) \Big| + \Big| \lambda(J_{\boldsymbol{t}}^*) \big(M - \frac{N}{\lambda(A)} \big)\Big| \nonumber\\
& \le &\Big|\sum_{n=0}^{M-1} 1_{J_{\boldsymbol{t}}^*}(\boldsymbol{x}_n) - M\lambda(J_{\boldsymbol{t}}^*) \Big| + \Big| M\lambda(A) - N \Big|\nonumber\\
 & \le &\Big|\sum_{n=0}^{M-1} 1_{J_{\boldsymbol{t}}^*}(\boldsymbol{x}_n) - M\lambda(J_{\boldsymbol{t}}^*) \Big| + \Big| M\lambda(A) - \sum_{n=0}^{M-1} 1_{A}(\boldsymbol{x}_n) \Big|\nonumber\\
 & \le & 2 \sup_{\boldsymbol{t}\in[0,1]^s}\Big|\sum_{n=0}^{M-1} 1_{J_{\boldsymbol{t}}^*}(\boldsymbol{x}_n) - M\lambda(J_{\boldsymbol{t}}^*) \Big|.
\end{eqnarray}
Let us associate with each $Q_i$,  random points $\boldsymbol{x}_i\in Q_i$ with probability distribution
\begin{equation*}
\mathbb{P}(\boldsymbol{x}_i\ \in V)=\frac{\lambda(V)}{\lambda(Q_i)}=M\lambda(V),
\end{equation*}
for all measurable sets $V\subseteq Q_i$. 

 It  follows from Lemma~\ref{J_Intersect}  that $J_{\boldsymbol{t}}^*$ at most intersect $l:=3s^{1/2}\mathcal{M}_AM^{1-1/s}$ sets $Q_i$.
Therefore, $J_{\boldsymbol{t}}^*$ is representable as the disjoint union of sets $Q_i$  entirely contained in $J_{\boldsymbol{t}}^*$ and the union of at most $l$ pieces  which are intersections of some sets $Q_i$  and $J_{\boldsymbol{t}}^*$, i.e.
\begin{equation*}
J_{\boldsymbol{t}}^*=\bigcup\limits_{i\in I}Q_i \cup \bigcup\limits_{i\in J}(Q_i \cap J_{\boldsymbol{t}}^*),
\end{equation*}where the index-set $J$ has cardinality at most $\lceil3s^{1/2}\mathcal{M}_AM^{1-1/s}\rceil$. Since for every $Q_i$, $\lambda(Q_i)=1/M$ and $Q_i$ contains exactly one element of $\{\boldsymbol{z}_1,\ldots,\boldsymbol{z}_N\}$,
the discrepancy  of $\bigcup_{i\in I}Q_i$ is zero. Therefore, it remains to investigate the discrepancy of $\bigcup_{i\in J}(Q_i \cap J_{\boldsymbol{t}}^*)$.

Since $\lambda(A)=C/L$ and $N\ge M({C}/{L}-{3s^{1/2}\mathcal{M}_A}{M^{-1/s}})$ by Lemma~\ref{BoundN}, we have $M\le 2LN/C$  for all $M>(6Ls^{1/2}\mathcal{M}_A/C)^s$.
Consequently,
$$l=3s^{1/2}\mathcal{M}_AM^{1-1/s}\le 3s^{1/2}(2L)^{1-1/s}C^{1/s-1}\mathcal{M}_A N^{1-\frac{1}{s}}=\Omega N^{1-1/s},$$
where  $\Omega=3s^{1/2}(2L)^{1-1/s}C^{1/s-1}\mathcal{M}_A$.

Let us define the random variable $\chi_i$ for $0\le i\le l-1$ as follows
\begin{equation*}\chi_i=\left\{\begin{array}{ll}
1,&  \mbox{if }\boldsymbol{z}_i\in Q_i \cap J_{\boldsymbol{t}}^*,\\
0,& \mbox{if }\boldsymbol{z}_i\notin Q_i \cap J_{\boldsymbol{t}}^*.\\
\end{array}
\right.\end{equation*}
By definition,
\begin{align}\label{Discrepancyre}
   \Big|\sum_{n=0}^{M-1} 1_{J_{\boldsymbol{t}}^*}(\boldsymbol{x}_n) - M\lambda(J_{\boldsymbol{t}}^*) \Big|
                               =\Big|\sum_{i=0}^{l-1}\chi_i-M\sum_{i=0}^{l-1}\lambda(Q_i \cap J_{\boldsymbol{t}}^*)\Big|.
\end{align}
Because of  $\mathbb{P}(\chi_i=1)=\lambda(Q_i \cap J_{\boldsymbol{t}}^*)/\lambda(Q_i)=M\lambda(Q_i\cap J_{\boldsymbol{t}}^*)$, we have
\begin{equation}\label{Expectation}
    \mathbb{E}\chi_i=M\lambda(Q_i \cap J_{\boldsymbol{t}}^*),
\end{equation}
where $\mathbb{E}(\cdot)$ denotes the expected value.
By~\eqref{Discrepancyre} and \eqref{Expectation},
\begin{equation}
    \Delta_N(J_t^*;\boldsymbol{z}_1,\ldots,\boldsymbol{z}_N)=\Big|\sum_{n=0}^{M-1} 1_{J_{\boldsymbol{t}}^*}(\boldsymbol{x}_n) - M\lambda(J_{\boldsymbol{t}}^*) \Big|=\Big| \sum_{i=0}^{l-1}(\chi_i-\mathbb{E}\chi_i)\Big|.
\end{equation}\label{discrExpect}

Since the random variables $\chi_i$ for $0\le i\le l-1$ are independent of each other, in order to estimate the sum $\sum_{i=0}^{l-1}(\chi_i-\mathbb{E}\chi_i)$ we are able to apply the classical Bernstein-Chernoff inequality of large deviation type.
Let $\sigma_i^2=\mathbb{E}(\chi_i-\mathbb{E}\chi_i)^2$ and set $\beta=(\sum_{i=1}^{l}\sigma_i^2)^{1/2}$. Let
\begin{equation*}
  \gamma=\theta l^{1/2}(\log N)^{1/2},
\end{equation*}
where $\theta$ is a constant  depending only on  the dimension $s$ which will be fixed later. Without loss of generality, assume that $N\ge 3$.

\textbf{Case 1:} If $\gamma\le \beta^2$,
 since $\beta^2\le l\le \Omega N^{1-\frac{1}{s}}$, by Lemma~\ref{Bernstein-Chernoff} we obtain
\begin{align}\label{Prob_expect1}
    &\mathbb{P}\left(\Delta_N(J_t^*;\boldsymbol{z}_1,\ldots,\boldsymbol{z}_N)\ge \theta l^{1/2}(\log N)^{1/2}\right)\nonumber \\
    &=\mathbb{P}\Big(\big| \sum_{i=1}^{l}(\chi_i-\mathbb{E}\chi_i)\big|\ge \gamma\Big) \le 2e^{-{\gamma^2}/(4\beta^2)}\le 2N^{-{\theta^2}/{4}}.
\end{align}
Though the class of axis-parallel boxes  is uncountable, it suffices to consider a small subclass. Based on the argument in Lemma~\ref{delta_coverOptA},
 there is a $1/M$-cover  of cardinality $(2e)^{s-1}(M+1)^{s-1}\le(2e)^{s-1}({2LN}/{C}+1)^{s-1}$ for $M>M_0$ such that  there exist $R_1,R_2\in \Gamma_{1/M}$ having the properties $R_1\subset J_{\boldsymbol{t}}^* \subset R_2$  and $\lambda(R_2\setminus R_1)\le 1/M$. From this it follows that
\begin{equation*}
     \Delta_N(J_{\boldsymbol{t}}^*;\boldsymbol{z}_1,\ldots,\boldsymbol{z}_N)\le \max_{i=1,2}\Delta(R_i;\boldsymbol{z}_1,\ldots,\boldsymbol{z}_N)+1,
\end{equation*} see, for instance, \cite[Lemma~3.1]{DG2005} and \cite[Section~2.1]{HNWW2001}.
This means that we can restrict ourselves to the elements of $\Gamma_{1/M}$.

In view of \eqref{Prob_expect1}
\begin{equation*}
 \mathbb{P}\big(\Delta(R_i;\boldsymbol{z}_1,\ldots,\boldsymbol{z}_N)\ge \gamma \big)
 \le |\Gamma_{1/M}|2N^{-\frac{\theta^2}{4}}
  \le 2N^{-\frac{\theta^2}{4}}(2e)^{s-1}\big(\frac{2LN}{C}+1\big)^{s-1}<1,
 \end{equation*}
 for $\theta= 2\sqrt{2s}$ and $N\ge \frac{8 e}{C} +2$.

 \textbf{Case 2:} On the other hand, if $\gamma\ge \beta^2$, then
  by Lemma~\ref{Bernstein-Chernoff} we obtain
\begin{align}\label{Prob_expect}
   &\mathbb{P}\Big( \Delta(J_t^*;\boldsymbol{z}_1,\ldots,\boldsymbol{z}_N)\ge \theta l^{1/2}(\log N)^{1/2}\Big)\nonumber\\
    &=\mathbb{P}\Big(\big| \sum_{i=1}^{l}(\chi_i-\mathbb{E}\chi_i)\big|\ge \gamma\Big) \le 2e^{-\frac{\theta l^{1/2}(\log N)^{1/2}}{4}}.
\end{align}

 Similarly, using the $1/M$-cover technique above, for $\theta=2\sqrt{2s}$ and sufficiently large $N$  we have
 \begin{align*}
 \mathbb{P}\big(\Delta(R_i;\boldsymbol{z}_1,\ldots,\boldsymbol{z}_N)\ge \gamma \big)
 &\le|\Gamma_{1/M}|2e^{-\frac{\theta l^{1/2}(\log N)^{1/2}}{4}}\\
 & \le 2e^{-\frac{\theta l^{1/2}(\log N)^{1/2}}{4}}(2e)^{s-1}\big(\frac{2LN}{C}+1\big)^{s-1}<1,
 \end{align*} where the last equation is satisfied for all large enough $N$.

By \eqref{Discre} and \eqref{DiscreTrans}, we obtain that, with positive probability, Algorithm~\ref{algorithmSAR} yields a  point set $Y_N^{(s-1)}$ such that
\begin{equation*}
    D_{N,\psi}^*(Y_N^{(s-1)})\le \sqrt{2s}\Omega^{1/2}N^{-\frac{1}{2}-\frac{1}{2s}}(\log N)^{1/2}+1/M.
\end{equation*}

By Lemma~\ref{NandM}, we have $1/M\le 2C/(LN)$  for sufficiently large $N$. Thus the proof of Theorem~\ref{Bound_Stratified} is complete.
\end{proof}
\subsection{Upper  Bound on the $L_q$-discrepancy}
In this section we prove an upper bound on the expected value of the $L_q$-discrepancy for $2\le q\le \infty$.
\begin{theorem}\label{varianceBound}  Let the unnormalized density  function $\psi:[0,1]^{s-1}\to \mathbb{R}_+$ satisfy all the assumptions stated in Theorem~\ref{Bound_Stratified}. Let $Y_N^{(s-1)}$ be the samples generated by the acceptance-rejection sampler using stratified inputs.  Then we have for $2\le q\le \infty$,
\begin{equation*}
 \big( \mathbb{E}[N^{q}L^{q}_{q,N}(Y_N^{(s-1)})]\big)^{1/q}\le (3s^{1/2}\mathcal{M}_A)^{1-1/q}(2LC^{-1})^{(1-1/s)(1-1/q)}N^{(1-1/s)(1-1/q)},
\end{equation*} where $\mathcal{M}_A$ is the $(s-1)-$dimensional Minkowski content and
 the expectation is taken with respect to the stratified inputs.
\end{theorem}

\begin{proof} Let $J_{\boldsymbol{t}}^*=([\boldsymbol{0},\boldsymbol{t})\times [0,1])\bigcap A $, where $\boldsymbol{t}=(t_1,\ldots,t_{s-1})\in[0,1]^{s-1}$.
 Let $$\xi_i(t)=1_{Q_i\cap J_{\boldsymbol{t}}^*}(\boldsymbol{x}_i)-\lambda(Q_i\cap J_{\boldsymbol{t}}^*)/\lambda(Q_i),$$ where $Q_i$ for $0\le i\le M-1$ is the covering of $[0,1]^s$ with $\lambda(Q_i)=1/M$. Then $\mathbb{E}(\xi_i(t))=0$ since we have $\mathbb{E}[1_{Q_i\cap J_{\boldsymbol{t}}^*}(\boldsymbol{x}_i)]=M\lambda(Q_i\cap J_{\boldsymbol{t}}^*)$.
  Hence
 \begin{eqnarray*}
\mathbb{E}[\xi_i^2(t)]&= &\mathbb{E}[(1_{Q_i\cap J_{\boldsymbol{t}}^*}(\boldsymbol{x}_i)-M\lambda(Q_i\cap J_{\boldsymbol{t}}^*))^2]\\
   &=&\mathbb{E}[1_{Q_i\cap J_{\boldsymbol{t}}^*}(\boldsymbol{x}_i)]-2M\lambda(Q_i\cap J_{\boldsymbol{t}}^*)\mathbb{E}[1_{Q_i\cap J_{\boldsymbol{t}}^*}(\boldsymbol{x}_i)]+M^2\lambda^2(Q_i\cap J_{\boldsymbol{t}}^*)\\
   &=&M\lambda(Q_i\cap J_{\boldsymbol{t}}^*)(1-M\lambda(Q_i\cap J_{\boldsymbol{t}}^*))\\
   &\le& M\lambda(Q_i\cap J_{\boldsymbol{t}}^*)\le 1.
 \end{eqnarray*}
 If $Q_i\subseteq J_t^*$ or if $Q_i\cap J_t^*=\emptyset$, we have $\xi_i(t)=0$.
  We order the sets $Q_i$ such that $Q_0, Q_1,\ldots, Q_{i_0}$ satisfy $Q_i\cap J_t^*\neq\emptyset$ and $Q_i\nsubseteq J_t^*$ (i.e. $Q_i$ intersects the boundary of $J_t^*$) and the remaining sets $Q_i$
  either satisfy $Q_i\cap J_t^*=\emptyset$ or $Q_i\subseteq J_t^*$.
  Due to the fact that the density curve $\psi$ at most intersects with $l:=3s^{1/2}\mathcal{M}(\partial A)M^{1-1/s}$ sets $Q_i$, if $\partial A$ admits an $(s-1)-$dimensional Minkowski content, it follows that, for $q=2$,
\begin{eqnarray*}
\big( \mathbb{E}[N^2L_{2,N}^2(Y_N^{(s-1)})]\big)^{1/2}&=& \Big(\mathbb{E}\Big[\int_{[0,1]^s}\big|\sum_{i=0}^{M-1}\xi_{i}(\boldsymbol{t}) \big|^2\D \boldsymbol{t}\Big]\Big)^{1/2}\\
 &=&\Big(\int_{[0,1]^s}\mathbb{E}\big[\sum_{i=0}^{M-1}\xi_{i}(\boldsymbol{t})\big]^2 \D \boldsymbol{t} \Big)^{1/2}\\
 &=&\Big(\int_{[0,1]^s}\sum_{i=0}^{l-1}\mathbb{E}[\xi_i(t)^2]\D \boldsymbol{t} \Big)^{1/2}\le l^{1/2}.
\end{eqnarray*}
Since $|\xi_i(t)|\le 1$, for $q=\infty$, we have
\begin{eqnarray*}
 \sup_{P_M\subset [0,1]^{s}}|ND_N^*(Y_N^{(s-1)})|&=& \sup_{P_M\subset [0,1]^{s}}\sup_{\boldsymbol{t}\in[0,1]^{s-1}} \big|\sum_{i=0}^{M-1}\xi_{i}(\boldsymbol{t}) \big|= \sup_{P_M\subset [0,1]^{s}}\sup_{\boldsymbol{t}\in[0,1]^{s-1}} \big|\sum_{i=0}^{l-1}\xi_{i}(\boldsymbol{t}) \big|\\
            &\le& \sup_{P_M\in [0,1]^{s}}\sup_{\boldsymbol{t}\in[0,1]^{s-1}} \sum_{i=0}^{l-1}\big|\xi_{i}(\boldsymbol{t}) \big|\le  l.
  \end{eqnarray*}
Therefore, for $2\le q\le \infty$,
\begin{equation*}
\big( \mathbb{E}[N^{q}L^{q}_{q,N}(Y_N^{(s-1)})]\big)^{1/q}\le  l^{1-1/q},
\end{equation*}
which is a consequence of the log-convexity of $L_p$-norms, i.e $\| f\|_{p_\theta}\le \| f\|_{p_0}^{1-\theta}\| f\|_{p_1}^{\theta}$, where $1/p_\theta=(1-\theta)/p_0+\theta/p_1$. In our case,
$p_0=2$ and $p_1=\infty$.

Additionally, following from  Lemma~\ref{BoundN}, we have $M\le 2LN/C$ whenever $M>(6Ls^{1/2}\mathcal{M}_A/C)^s$.
Hence we obtain the desired result by substituting $l=3s^{1/2}\mathcal{M}_AM^{1-1/s}$ and replacing $M$ in terms of $N$.
\end{proof}
\begin{remark}It would also be interesting to obtain an upper bound for  $1\le q<2$. See Heinrich \cite{Heinrich2000} for a possible proof technique.
We leave it as an open problem.
\end{remark}
\section{Improved Rate of Convergence for Deterministic Acceptance-Rejection Sampler}
\label{Sec:4}
In this section, we  prove a convergence rate of order $N^{-\alpha}$ for $1/s \le \alpha < 1$, where $\alpha$ depends on the target density $\psi$. See Corollary~\ref{Maincorollary} below for details. For
this result we use $(t,m,s)$-nets (see Definition~\ref{DeftmsNet} below) as inputs instead of stratified samples. The value of $\alpha$ here depends on how well the graph of $\psi$ can be covered by certain rectangles (see Equation~\eqref{G_k}). In practice this covering rate of order $N^{-\alpha}$ is hard to determine precisely, where $\alpha$ can range anywhere from $1/s $ to $ < 1$, where $\alpha$ arbitrarily close to $1$ can be achieved if $\psi$ is constant. We also provide a simple example in dimension $s=2$ for which $\alpha$ can take on the values $\alpha = 1- \ell^{-1}$ for $\ell \in \mathbb{N}$, $\ell \ge 2$. See Example~\ref{Example-Section2} for details.

We first establish some notation and some useful
definitions and then obtain theoretical results.
First we introduce the definition of $(t,m,s)$-nets  in base $b$ (see
\cite{DP10}) which we use as the driver sequence.
The following fundamental definitions of elementary interval and fair
sets are used to define a $(t,m,s)$-net in base $b$.
\begin{definition}\label{Einterval}[{b-adic elementary interval}] Let $b\geq2$ be an
integer. An $s$-di\-mensional $b$-adic elementary interval is an interval
of the form
\begin{equation*}
\prod_{i=1}^s\left[\displaystyle\frac{a_i}{b^{d_i}},\displaystyle\frac
{a_i+1}{b^{d_i}}\right)
\end{equation*}
with integers $0\leq a_i<b^{d_i}$ and $d_i\geq0$ for all $1\leq i\leq
s$. If $d_1,\ldots,d_s$ are such that $d_1+\cdots+d_s=k$, then we say
that the elementary interval is of order $k$.
\end{definition}

\begin{definition}[{fair sets}]For a given set $P_N=\{\boldsymbol
{x}_0,\boldsymbol{x}_1,\ldots,\boldsymbol{x}_{N-1}\}$ consisting of $N$
points in $[0,1)^s$, we say for a subset $J$ of $[0,1)^s$ to be fair
with respect to $P_N$, if
\begin{equation*}
\displaystyle\frac{1}{N}\sum_{n=0}^{N-1}1_J(\boldsymbol{x}_n)=\lambda(J),
\end{equation*}
where $1_J(\boldsymbol{x}_n)$ is the indicator function of the set $J$.
\end{definition}

\begin{definition}[{$(t,m,s)$-nets in base b}]\label{DeftmsNet}
For a given dimension $s\geq1$, an integer base $b\geq2$, a positive
integer $m$ and an integer $t$ with $0\leq t\leq m$,
a point set $Q_{m,s}$ of $b^m$ points in $[0,1)^s$ is called a
$(t,m,s)$-nets in base $b$ if the point set $Q_{m,s}$ is fair with
respect to all b-adic s-dimensional elementary intervals of order at
most $m-t$.
\end{definition}

We present the acceptance-rejection algorithm using $(t,m,s)$-nets as driver sequence.
\begin{algorithm}\label{algorithmDAR}  Let the target density $\psi:[0,1]^{s-1}\to \mathbb{R}_+$, where $s\geq 2$, be given. Assume that there exists a constant $L<\infty$ such that $\psi(\boldsymbol{x}) \leq L$ for  all $\boldsymbol{x}\in[0,1]^{s-1}$. Let $A=\{\boldsymbol{z}\in[0,1]^s:\psi(z_1,\ldots,z_{s-1})\geq L x_s\}$. Suppose we aim to obtain approximately $N$ samples from $\psi$.
\begin{itemize}
  \item [i)~] Let $M=b^m\ge\left \lceil N/(\int_{[0,1]^{s-1}}\psi(\boldsymbol{x})/L d\boldsymbol{x})\right\rceil$, where $m\in\mathbb{N}$ is the smallest integer satisfying this inequality.  Generate a $(t,m,s)$-net $Q_{m,s}=\{\boldsymbol{x}_0,\boldsymbol{x}_1,\ldots,\boldsymbol{x}_{b^m-1}\}$ in base $b$.
  \item [ii)~]Use the acceptance-rejection method for the points $Q_{m,s}$ with respect to the density $\psi$, i.e. we accept the point $\boldsymbol{x}_n$ if $\boldsymbol{x}_n \in A$, otherwise reject. Let $P_N^{(s)}=A\cap Q_{m,s}=\{\boldsymbol{z}_0,\ldots, \boldsymbol{z}_{N-1}\}$ be the sample set we accept.
  \item [iii)~]  Project the points  $P_N^{(s)}$ onto  the first $(s-1)$ coordinates.
  Let $Y_N^{(s-1)}=\{\boldsymbol{y}_0,\ldots, \boldsymbol{y}_{N-1}\}\subseteq [0,1]^{s-1}$ be the projections of the points  $P_N^{(s)}$.
 \item [iv)~] Return the point set $Y_N^{(s-1)}$.
\end{itemize}
\end{algorithm}

In the following we show that an improvement of the discrepancy bound for the deterministic acceptance-rejection sampler  is possible. Let an unnormalized  density function $\psi:[0,1]^{s-1}\to \mathbb{R}_+$, with $s\geq 2$, be given. Let again
$$A=\{\boldsymbol{z}=(z_1,\ldots,z_{s})\in[0,1]^s:\psi(z_1,\ldots,z_{s-1})\geq L z_s\}$$ and $J_{\boldsymbol{t}}^*=([\boldsymbol{0},\boldsymbol{t})\times [0,1])\bigcap A $. Let $\partial J^*_{\bst}$ denote the boundary of $J^*_{\bst}$ and  $\partial [0,1]^s$ denotes the boundary of $[0,1]^s$. For $k\in \mathbb{N}$  we define the covering number
\begin{align}\label{G_k}
\Gamma_k(\psi) = \sup_{\bst \in [0,1]^s} \min \{v: & \exists U_{1},\ldots, U_v \in \mathcal{E}_k:  ( \partial J^*_{\bst} \setminus \partial [0,1]^s ) \subseteq \bigcup_{i=1}^v U_i, \nonumber \\ &  U_i \cap U_{i'} = \emptyset \mbox{ for } 1 \le i < i' \le v  \},
\end{align}
where $\mathcal{E}_k$ is the family of elementary intervals of order $k$.
\begin{lemma} \label{Mainlemma}
Let $\psi:[0,1]^{s-1} \to [0,1]$ be an unnormalized target density and let the covering number $\Gamma_{m-t}(\psi)$ be given by \eqref{G_k}. Then the  discrepancy of the point set $Y_N^{(s-1)}=\{\boldsymbol{y}_0,\boldsymbol{y}_1,\ldots, \boldsymbol{y}_{N-1}\} \subseteq [0,1]^{s-1}$ generated by
Algorithm~\ref{algorithmDAR} using a $(t,m,s)$-net in base $b$, for large enough $N$, satisfies
\begin{equation*}
   D_{N,\psi}^*(Y_N^{(s-1)})
    \le 4 C^{-1} b^t \Gamma_{m-t}(\psi) N^{-1},
  \end{equation*}
  where $C=\int_{[0,1]^{s-1}}\psi({\boldsymbol{z}})d\boldsymbol{z}$.
\end{lemma}

\begin{proof}
Let $\bst \in [0,1]^s$ be given. Let $v = \Gamma_{m-t}(\psi)$ and $U_1, \ldots, U_v$ be elementary intervals of order $m-t$ such that $U_1 \cup U_2 \cup \cdots \cup U_v \supseteq (\partial J^*_{\bst} \setminus \partial [0,1]^s)$ and $U_i \cap U_{i'} = \emptyset$ for $1 \le i < i' \le v$. Let $V_1, \ldots, V_z \in \mathcal{E}_{m-t}$ with $V_i \subseteq J^*_{\boldsymbol{t}}$, $V_i \cap V_{i'} = \emptyset$ for all $1 \le i < i' \le z$ and $V_i \cap U_i = \emptyset$ such that $\bigcup_{i=1}^z V_i \cup \bigcup_{i=1}^v U_i \supseteq J^*_{\bst}$. We define
\begin{equation*}
\overline{W} = \bigcup_{i=1}^z V_i \cup \bigcup_{i=1}^v U_i
\end{equation*}
and
\begin{equation*}
W^o = \bigcup_{i=1}^z V_i.
\end{equation*}
Then $\overline{W}$ and $W^o$ are fair with respect to the $(t,m,s)$-net, $W^o \subseteq J^*_{\boldsymbol{t}} \subseteq \overline{W}$ and
\begin{equation*}
\lambda(\overline{W} \setminus J^*_{\boldsymbol{t}}), \lambda(J^*_{\boldsymbol{t}} \setminus W^o) \le \lambda(\overline{W} \setminus W^o) = \sum_{i=1}^v \lambda(U_i) = \sum_{i=1}^v b^{-m+t} = b^{-m+t} \Gamma_{m-t}(\psi).
\end{equation*}
The proof of the result now follows by the same arguments as the proofs in \cite[Lemma~1\&Theorem~1]{ZD2014}.
\end{proof}

 From Lemma~\ref{J_Intersect} we have that if $\partial A$ admits an $(s-1)-$dimensional Minkowski content, then
  \begin{equation*}
  \Gamma_{k}(\psi) \le c_s b^{(1-1/s)k}.
\end{equation*}
 This yields a convergence rate of order $N^{-1/s}$ in Lemma~\ref{Mainlemma}. Another known example is the following. Assume that $\psi$ is  constant. Since the graph of $\psi$ can be covered by just one elementary interval of order $m-t$, this is the simplest possible case. The results from \cite[Section~3]{Niederriter1987} (see also \cite[p. 184--190]{DP10} for an exposition in dimensions $s=1,2,3$) imply that $\Gamma_{k}(\psi) \le C_s k^{s-1}$ for some constant $C_s$ which depends only on $s$. This yields the convergence rate of order $(\log N)^{s-1} N^{-1}$ in Lemma~\ref{Mainlemma}. Thus, in general, there are constants $c_{s, \psi}$ and $C_{s, \psi}$ depending only on $s$ and $\psi$ such that
\begin{equation}\label{eq_Gamma_psi}
c_{s,\psi} k^{s-1} \le \Gamma_{k}(\psi) \le C_{s, \psi} b^{(1-1/s) k},
\end{equation}
whenever the set $\partial A$ admits an $(s-1)-$dimensional Minkowski content. This yields a convergence rate in Lemma~\ref{Mainlemma} of order $N^{-\alpha}$ with $1/s \le \alpha < 1$, where the precise value of $\alpha$ depends on $\psi$. We obtain the following corollary.

\begin{corollary} \label{Maincorollary}
Let $\psi:[0,1]^{s-1} \to [0,1]$ be an unnormalized target density and let $\Gamma_{k}(\psi)$ be given by \eqref{G_k}. Assume that there is a constant $\Theta > 0$ such that
\begin{equation*}
\Gamma_{k}(\psi) \le \Theta b^{(1-\alpha) k} k^{\beta} \quad\mbox{for all } k \in \mathbb{N},
\end{equation*}
for some $1/s \le \alpha \le 1$ and $\beta \ge 0$. Then there is a constant $\Delta_{s,t,\psi} > 0$ which depends only on $s,t$ and $\psi$, such that the  discrepancy of the point set $Y_N^{(s-1)}=\{\boldsymbol{y}_0,\boldsymbol{y}_1,\ldots, \boldsymbol{y}_{N-1}\} \subseteq [0,1]^{s-1}$ generated by
Algorithm \ref{algorithmDAR} using a $(t,m,s)$-net in base $b$, for large enough $N$, satisfies
\begin{equation*}
  D_{N,\psi}^*(Y_N^{(s-1)})
    \le \Delta_{s,t,\psi} N^{-\alpha} (\log N)^{\beta}.
  \end{equation*}
\end{corollary}


\begin{example}\label{Example-Section2}
To illustrate the bound in Corollary~\ref{Maincorollary}, we consider now an example for which we can obtain an explicit bound on $\Gamma_k(\psi)$ of order $b^{k(1-\alpha)}$ for $1/2 \le \alpha < 1$. For simplicity let $s=2$ and $\alpha = 1- \ell^{-1}$ for some $\ell \in \mathbb{N}$ with $\ell \ge 2$. We define now a function $\psi_\ell: [0,1] \to [0,1]$ in the following way: let $x \in [0,1)$ have $b$-adic expansion
\begin{equation*}
x = \frac{\xi_1}{b} + \frac{\xi_2}{b^2} + \frac{\xi_3}{b^3} + \cdots
\end{equation*}
where $\xi_i \in \{0,1,\ldots, b-1\}$ and  assume that infinitely many of the $\xi_i$ are different from $b-1$. Then set
\begin{equation*}
\psi_\ell(x) = \frac{\xi_1}{b^{l-1}} + \frac{\xi_2}{b^{2(l-1)}} + \frac{\xi_3}{b^{3(l-1)}} + \cdots.
\end{equation*}
Let $t \in [0,1)$. In the following we define elementary intervals of order $k \in \mathbb{N}$ which cover $\partial J_t^* \setminus \partial [0,1]^2$. Assume first that $k$ is a multiple of $\ell$, then let $g = k/\ell$. Then we define the following elementary intervals of order $k = g \ell$:
\begin{align}\label{cover_int}
& \left[\frac{a_1}{b} + \cdots + \frac{a_{g-1}}{b^{g-1}} + \frac{a_g}{b^g}, \frac{a_1}{b} + \cdots + \frac{a_{g-1}}{b^{g-1} } + \frac{a_g + 1}{b^g} \right) \times \nonumber \\ & \left[\frac{a_1}{b^{\ell-1}} + \cdots + \frac{a_{g-1} }{b^{(g-1) (\ell-1)} } + \frac{a_g}{b^{g(\ell-1)} }, \frac{a_1}{b^{\ell-1}} + \cdots + \frac{a_{g-1}}{b^{(g-1) (\ell-1)} } + \frac{a_g+1}{b^{g(\ell-1)} }\right),
\end{align}
where $a_1, \ldots,  a_g \in \{0,1, \ldots, b-1\}$ run through all possible choices such that
\begin{equation*}
 \frac{a_1}{b} + \cdots + \frac{a_{g-1}}{b^{g-1} } + \frac{a_g + 1}{b^g} \le t.
\end{equation*}
The number of these choices for $a_1, \ldots, a_g$ is bounded by $b^g$. Let
\begin{equation*}
t = \frac{t_1}{b} + \cdots + \frac{t_g}{b^g} + \frac{t_{g+1}}{b^{g+1}} + \cdots.
\end{equation*}
For integers $1 \le u \le g(\ell-1)$ and $0 \le c_u < t_{g+u}$, we define the intervals
\begin{align}\label{cover_int2}
& \left[\frac{t_1}{b} + \cdots + \frac{t_{g+ u-1}}{b^{g+u-1}} + \frac{c_u}{b^{g+u}}, \frac{t_1}{b} + \cdots + \frac{t_{g+u-1}}{b^{g+u-1} } + \frac{c_u + 1}{b^{g+u}} \right) \times \nonumber \\ & \left[\frac{d_1}{b} + \cdots + \frac{d_{g(\ell-1)-u}}{b^{g(\ell-1)-u}}, \frac{d_1}{b} + \cdots + \frac{d_{g(\ell-1)-u} }{b^{g(\ell-1)-u}}  + \frac{1}{b^{g(\ell-1)-u}} \right),
\end{align}
where $d_i = 0$ if $\ell \nmid i$, $d_i = t_{i/\ell}$ if $\ell | i$ and we set $\frac{d_1}{b} + \cdots + \frac{d_{g(\ell-1)-u}}{b^{g(\ell-1)-u}} = 0$ if $u = g(\ell-1)$. Further we define the interval
\begin{equation}\label{cover_int3}
\left[\frac{t_1}{b} + \cdots + \frac{t_{g \ell}}{b^{g \ell}}, \frac{t_1}{b} + \cdots + \frac{t_{g\ell}}{b^{g\ell} } + \frac{1}{b^{g \ell}} \right) \times [0,1).
\end{equation}
The intervals defined in \eqref{cover_int}, \eqref{cover_int2} and \eqref{cover_int3} cover $\partial J^*_t \setminus \partial [0,1]^2$. Thus we have
\begin{equation*}
\Gamma_{g \ell}(\psi_\ell) \le b^{g} + b g(\ell-1) + 1 \le \ell b^g.
\end{equation*}
For arbitrary $k \in \mathbb{N}$ we can use elementary intervals of order $k$ which cover the same area as the intervals \eqref{cover_int}, \eqref{cover_int2} and \eqref{cover_int3}. Thus we have at most $b^{\ell - 1}$ times as many intervals and we therefore obtain
\begin{equation*}
\Gamma_k(\psi_\ell) \le \ell b^{k/\ell + \ell-1}.
\end{equation*}
Thus we obtain
\begin{equation*}
   \sup_{\boldsymbol{t}\in[0,1]} \left|\frac{1}{N}\sum_{n=0}^{N-1}1_{[0, t)}({y_n})
    -  \frac{1}{C}\int_0^t \psi_\ell(z) \D z \right|
    \le \Delta_{s,t,\psi} N^{-(1-\frac{1}{\ell}) }.
  \end{equation*}
\end{example}

\begin{remark} In order to obtain similar results as in this section for stratified inputs rather than $(t,m,s)-$nets, one would have to use the elementary intervals $U_1,\ldots,U_v$ of order $k$ which yield  a covering of $\partial J_t^*\setminus \partial [0,1]^s$ for all $\boldsymbol{t}\in [0,1]^{s-1}$. From this covering one would then have to construct a covering of $\partial A\setminus \partial [0,1]^s$ and use this covering to obtain stratified  inputs. Since such a covering is not easily available  in general, we did not pursue this approach further.
\end{remark}

%
\begin{acknowledgement}
\end{acknowledgement}
H. Zhu was supported by a PhD scholarship from the
University of New South Wales. J.~Dick was supported by a Queen
Elizabeth 2 Fellowship from the Australian Research Council.


\begin{thebibliography}{99}

\bibitem{ACV2008}
L.~Ambrosio, A.~Colesanti and E.~Villa.
\newblock Outer Minkowski content for some classes of closed sets.
\newblock {\em Mathematische Annalen}, 342, 727--748, 2008.

\bibitem{Beck1984}
J.~Beck.
\newblock Some upper bounds in the theory of irregularities of distribution.
\newblock {\em Acta Arithmetica}, 43, 115--130, 1984.

\bibitem{BHL2013}
C.~Botts, W.~H\"ormann and J.~ Leydold.
\newblock Transformed density rejection with inflection points.
\newblock{\em  Statatistics and Computing}, 23, 251-260, 2013.

%
\bibitem{ChenThesis2011}
S. Chen.
\newblock {\em Consistency and convergence rate of Markov  chain quasi Monte Carlo with examples}.
\newblock PhD thesis, Stanford University, 2011.

\bibitem{CDO2011}
S.~Chen, J.~Dick and A.B.~Owen.
\newblock Consistency of Markov chain quasi-Monte Carlo on continuous state spaces.
\newblock{\em Annals of Statistics}, 39, 673--701, 2011.


\bibitem{Devroye1984}
L.~Devroye.
\newblock A simple algorithm  for generating random variates with a log-concave density.
\newblock {\em Computing}, 33, 247-257, 1984.


\bibitem{Devroye1986}
L.~Devroye.
\newblock {\em Nonuniform Random Variate Generation}.
\newblock Springer-Verlag, New York, 1986.


\bibitem{DP10}
J.~Dick and F.~Pillichshammer.
\newblock {\em Digital Nets and Sequences: Discrepancy Theory and Quasi-Monte
  Carlo Integration}.
\newblock Cambridge University Press, 2010.

\bibitem{DR2013}
J.~Dick and D.~Rudolf.
\newblock{Discrepancy estimates for variance bounding
Markov chain quasi-Monte Carlo}.
\newblock Available at {http://arxiv.org/abs/1311.1890[stat.CO]},
submitted, 2013.

\bibitem{DRZ2013}
J.~Dick, D.~Rudolf and H.~Zhu.
\newblock{Discrepancy bounds for uniformly ergodic Markov chain quasi-Monte Carlo}.
\newblock Available at {http://arxiv.org/abs/1303.2423 [stat.CO]},
submitted, 2013.

\bibitem{DG2005}
B.~Doerr, M.~Gnewuch and A.~Srivastav.
\newblock { Bounds and constructions for the star-discrepancy via $\delta$-covers}.
\newblock {\em Journal of Complexity}, 21, 691--709, 2005.



\bibitem{GC2014}
M. Gerber and N.~Chopin.
\newblock {Sequential quasi-Monte Carlo}.
\newblock Available at {http://arxiv.org/abs/1402.4039}{1402.4039 [stat.CO]}, 2014.

\bibitem{G2008}
M.~Gnewuch.
\newblock  Bracketing number for axis-parallel boxes and application to geometric discrepancy.
\newblock {\em Journal of Complexity}, 24, 154--172, 2008.

\bibitem{HeOwen2014}
Z.~He and A.B.~Owen.
\newblock{Extensible grids: uniform sampling on a space-filling curve}.
\newblock Available at {http://arxiv.org/abs/1406.4549 [stat.ME]}, 2014.

\bibitem{Heinrich2000}
S.~Heinrich.
\newblock The multilevel method of dependent tests.
\newblock In N.~Balakrishnan, V.B.~Melas, S.M.~Ermakov, editors,
{\em Advances in Stochastic Simulation Methods},  pages 47--62.
Birkh\"auser, 2000.

\bibitem{HNWW2001}
S.~Heinrich, E.~Novak, G.W.~Wasilkowski and H.~Wo\'{z}niakowski.
\newblock{ The inverse of the star-discrepancy depends linearly on the dimension}.
\newblock {\em Acta Arithmetica}, 96, 279--302, 2001.

\bibitem{Hormann1995}
W.~H\"ormann.
\newblock{ A reject technique  for sampling from T-concave distributions}.
\newblock{\em ACM Transactions on Mathematical Software},
21, 182-193, 1995.



\bibitem{HLD04}
W.~H\"ormann, J.~Leydold and G.~Derflinger.
\newblock {\em Automatic Nonuniform Random Variate Generation}.
\newblock Springer-Verlag, Berlin, 2004.


\bibitem{MC95}
W.J.~Morokoff and R.E.~Caflisch.
\newblock {Quasi-Monte Carlo integration}.
\newblock {\em Journal of Computational Physics}, 122, 218--230, 1995.



\bibitem{MC96}
B.~Moskowitz and R.E.~Caflisch.
\newblock Smoothness and dimension reduction in quasi-Monte Carlo methods.
\newblock {\em Mathematical and Computer Modelling}, 23, 37--54, 1996.


\bibitem{Pierre2008}
P.~L'Ecuyer, C.~Lecot and B.~Tuffin.
\newblock{A randomized quasi-Monte Carlo simulation method for Markov chains}.
\newblock{\em Operation Research}, 56, 958--975, 2008.

\bibitem{NO2013}
N.~Nguyen and G.~\"{O}kten.
\newblock {The acceptance-rejection method for low discrepancy sequences}.
\newblock  Available at {http://arxiv.org/abs/1403.5599 [q-fin.CP]}, 2014.


\bibitem{Niederriter1974}
L. Kuipers and H. Niederreiter.
\newblock{\em Uniform Distribution of Sequences}.
\newblock John Wiley, New York, 1974.

\bibitem{Niederriter1987}
H. Niederreiter.
 \newblock Point sets and sequences with small discrepancy.
 \newblock {\em Monatshefte f\"ur Mathematik}, 104, 273--337, 1987.


\bibitem{NW75}
H.~Niederreiter,~H and J.M.~Wills.
\newblock Diskrepanz und Distanz von Ma{\ss}en bez\"uglich konvexer und Jordanscher Mengen (German).
\newblock {\em Mathematische Zeitschrift}, 144, 125--134, 1975.

\bibitem{OwenbookDraft}
A.B.~Owen.
\newblock{Monte Carlo theory, methods and examples}.
\newblock Available at \url{http://www-stat.stanford.edu/~owen/mc/}.
Last accessed on 7 July 2014.


\bibitem{RobertCasella2004}
C.~Robert and G.~Casella,
\newblock {\em Monte Carlo Statistical Methods}.
\newblock  Springer-Verlag, New York, second edition, 2004.


\bibitem{TribbleThesis2007}
S.D.~Tribble.
\newblock{\em Markov chain Monte Carlo algorithms using completely uniformly distributed driving sequences}.
\newblock PhD thesis, Stanford University, 2007.


\bibitem{TribbleOwen2005}
S.D.~Tribble and A.B.~Owen,
\newblock { Constructions of weakly CUD sequences for MCMC}.
\newblock {\em Electronic Journal of Statistics}, 2, 634--660, 2008.

\bibitem{Wang1999}
X.~Wang.
\newblock{Quasi-Monte Carlo integration of characteristic functions and
    the rejection sampling method}.
\newblock{\em Computer Physics Communication}, 123, 16--26, 1999.

\bibitem{Wang00}
X.~Wang.
\newblock{Improving the rejection sampling method in quasi-Monte Carlo methods}.
\newblock{\em Journal of Computational and Applied Mathematics}, 114, 231--246, 2000.

\bibitem{ZD2014}
H.~Zhu and J.~Dick.
 \newblock Discrepancy bounds for deterministic acceptance-rejection samplers.
 \newblock {\em Electronic Journal of Statistics}, 8, 678-707, 2014.


\end{thebibliography}
\end{document}